\documentclass[letterpaper, 10 pt, conference]{ieeeconf}  
\IEEEoverridecommandlockouts 
\pdfoutput=1

\usepackage[left=1.62cm,right=1.62cm,top=1.91cm, bottom=2.38cm]{geometry}

\usepackage{amsmath,amssymb,color,cite}
\usepackage{graphicx}

\usepackage{latexsym}
\usepackage{algorithm}
\usepackage{algpseudocode}
\usepackage{verbatim}
\usepackage{cite}
\usepackage[english]{babel}
\usepackage{balance}
\usepackage{bbm}
\usepackage{hyperref}
\hypersetup{colorlinks=,linkcolor=,urlcolor=}

\let\labelindent\relax
\usepackage{amsmath,amsthm}
\usepackage{enumitem}

\usepackage{xcolor}

\usepackage{pgf}
\usepackage{pgfplots,tikz}
\usetikzlibrary{tikzmark}
\usepackage{arydshln}

\usepackage{amsfonts}
\usepackage{amssymb}
\usepackage{amsbsy}
\usepackage{bm}
\usepackage{mathtools}
\usepackage[font=footnotesize]{subcaption}
\usepackage[font=footnotesize]{caption}
\theoremstyle{definition}

\theoremstyle{plain}
\newtheorem{assumption}{Assumption}

\newtheorem{theorem}{Theorem}

\newtheorem{Remark}{Remark}[section]
\newcommand{\bRemark}{ \begin{Remark} \rm }
\newcommand{\eRemark}{ \end{Remark}    }
\newcommand{\hh}{\boldsymbol{h}}
\newcommand{\indi}{\mathbbm{1}}

\newcommand{\Scal}{\mathcal{S}}
\newcommand{\Mcal}{\mathcal{M}}
\newcommand{\Acal}{\mathcal{A}}
\newcommand{\Hcal}{\mathcal{H}}
\newcommand{\EE}{\mathbb{E}}
\newcommand{\PP}{\mathbb{P}}
\newcommand{\pistar}{{\pi^*}}
\newcommand{\Vstar}{V^*}

\newcommand{\pithresh}{\pi_{t,f}}

\title{\LARGE \bf \centering Optimal Control for Remote Patient Monitoring with \\ Multidimensional Health States

}

\author{Siddharth Chandak$^{1\dagger}$, Isha Thapa$^{2\dagger}$, Nicholas Bambos$^{1,2}$ and David Scheinker$^{2,3}$\\ 
$^{1}$ Department of Electrical Engineering, Stanford University, USA.\\
$^{2}$ Department of Management Science \& Engineering, Stanford University, USA.\\
$^{3}$ School of Medicine, Stanford University, USA.\\
{ \{chandaks, ishadt, bambos, dscheink\}@stanford.edu}
\thanks{$^\dagger$ Equal Contribution. Listed alphabetically.}
}

\begin{document}
\maketitle

\begin{abstract}

 Selecting the right monitoring level in Remote Patient Monitoring (RPM) systems for e-healthcare is crucial for balancing patient outcomes, various resources, and patient's quality of life. A prior work has used one-dimensional health representations, but patient health is inherently multidimensional and typically consists of many measurable physiological factors. In this paper, we introduce a multidimensional health state model within the RPM framework and use dynamic programming to study optimal monitoring strategies. Our analysis reveals that the optimal control is characterized by switching curves (for two-dimensional health states) or switching hyper-surfaces (in general): patients switch to intensive monitoring when health measurements cross a specific multidimensional surface. We further study how the optimal switching curve varies for different medical conditions and model parameters.
This finding of the optimal control structure provides actionable insights for clinicians and aids in resource planning. The tunable modeling framework enhances the applicability and effectiveness of RPM services across various medical conditions.

\end{abstract}


\section{Introduction}
Remote Patient Monitoring (RPM) is becoming an increasingly integral part of modern healthcare, enabling continuous observation of patients within their daily environments and enhancing both the quality of life and the level of delivered healthcare \cite{Farias, Malasinghe, Zinzuwadia}. Advances in wearable medical devices such as continuous glucose monitors (CGMs) \cite{Lee} and smartwatches which monitor vital signs \cite{Masoumian, BarDavid} among others have allowed for collection and transmission of health data for real-time analysis and timely medical interventions.

A significant challenge in RPM is determining the optimal intensity of monitoring. While intensive monitoring can lead to earlier detection and a prompt response, it can also cause patient stress and treatment fatigue due to its invasive nature \cite{heckman}, drain device battery faster \cite{BarDavid}, and increase other costs.
On the other hand, less intensive monitoring might not be sufficient for timely interventions. In a recent work \cite{chandak2024tiered}, we studied this trade-off using a one-dimensional health state model and showed that threshold policies, where patients switch to intensive monitoring when their health falls below a certain threshold, are optimal. However, in reality, a patient's health state is inherently multidimensional, where the dimensions represent different health measurements that clinicians use to make decisions \cite{Zinzuwadia,prahalad, senanayake, torous-suicide}.

Given the practical limitations of one-dimensional health states, in this paper, we extend our initial work \cite{chandak2024tiered} by introducing a multidimensional representation of health states within the RPM framework. We develop a dynamic programming approach to determine optimal monitoring strategies in this more complex setting. Our analysis reveals that controls characterized by \textit{switching curves} (in two dimensions) or \textit{hyper-surface} (in general) are optimal when considering multidimensional health states. That is, there exist switching curves/hyper-surfaces within the health space such that patients transition to intensive monitoring when their health indicators fall below this switching surface and return to ordinary monitoring when they improve. This finding provides clinicians with actionable guidelines for adjusting monitoring intensity based on a  multidimensional view of patient health.

The implications of the model and its analysis are significant for both patient care and resource management. By incorporating tunable parameters such as health improvement probabilities, monitoring options, and invasiveness costs, the framework can be adapted to specific medical conditions and tailored to individual patient needs. Furthermore, although not explicitly modeled, our approach allows for estimating the resources required to effectively manage patient populations, given the known time commitments associated with different levels of care. The rest of the paper is organized as follows. In Section II, we present the multidimensional RPM service model and describe the evolution of the patient's health state under different monitoring strategies. Section III discusses examples of critical health states and their impact on the optimal policy. Section IV provides numerical investigations of the optimal monitoring control, and Section V concludes the paper with potential extensions.

\section{The Remote Patient Monitoring Model} 
\label{sec:model}

Consider a patient whose health condition is modeled as an $n$-dimensional \textbf{health state}
$\hh=(h^{(1)},\ldots,h^{(n)})\in\Hcal\coloneqq \{0,1,\ldots,H\}^n$. At each time $t\in\{0,1,\ldots\}$, the patient's health state is given by $\hh_t\in\Hcal$. 
 The remote patient monitoring (RPM) service places the patient in a \textbf{monitoring state} $m_t\in\Mcal=\{o,i\}$ in each time period $t$ where $o$ denotes ordinary monitoring and $i$ 
 intensive monitoring.
Thus, one can view the joint
monitoring and patient state $s_t\coloneqq(m_t, \hh_t)\in\Mcal\times\Hcal\eqqcolon\Scal$ as the service state at time $t$.

The $n$ dimensions of the patient's health state correspond to different \textbf{health measurements}  monitored by the e-health service. For example, a program for Type 1 Diabetes management may have patients wear CGMs and examine multiple measurements. These measurements could include time in range (percentage of time glucose levels remained between 70 and 180mg/dl) and time with clinically significant hypoglycemia (percentage of time glucose levels were below 54mg/dl) \cite{prahalad}.
A higher patient health state corresponds to the patient having better health\footnote{In our model, continuous health measurements are discretized into clinically relevant categories. Additionally, in cases where increase in standard health measurement corresponds to worsening health, we use the inverse measurement instead. For example, the standard measurement is time with clinically significant hypoglycemia, but for our model we would use time without clinically significant hypoglycemia.}. Here `higher' is defined component-wise, i.e., for two health states $\hh$ and $\hh'$, with $h^{(j)} = h'^{(j)}$ for all $j\neq k$ and $h^{(k)}>h'^{(k)}$, the patient is said to have better health in state $\hh$ than in state $\hh'$. 

In particular, there exist {\bf critical} health states in the sense that, when the patient drifts into those `worse' states, they go beyond the scope of the current e-health service; and other emergency and/or more severe medical interventions are required, which are outside the scope of this service. For example, a person with diabetes experiencing a severe hypoglycemic episode may have to go to the emergency department or hospital \cite{mccoy-hypo}.  When the patient enters a critical health state under any monitoring state $i$ or $o$, the service evolution stops, as other medical measures/interventions are initiated. We denote the set of these critical health states by $\Hcal_C$. A simple example of these critical sets could just be the origin or a hypercube around origin. We discuss more such critical sets, with their medical relevance, in the next section. The origin $(0,\ldots,0)$ is always a part of $\Hcal_C$.

An advantage of our model is that the costs can be interpreted from multiple perspectives. In this section, we define the various costs based on the 
the patient's quality of life. Under ordinary monitoring, the patient incurs a constant cost $C_o\geq0$ at any state $(o,h)$ with $\hh\in\Hcal$. Similarly, under intensive monitoring, the patient incurs a constant cost $C_i\geq0$ at any state $(i,h)$ with $\hh\in\Hcal$. These costs reflect the invasiveness of the monitoring process to the patient's everyday lifestyle and quality of life. Of special interest are the critical health states, where  a cost of $C_c$ is incurred, and the model ceases to apply.

\subsection{Markov Evolution}
We model the system as a controlled Markov chain. Such models are often used for medical decision making \cite{alagoz}. We try to stay as simple as possible, while still capturing the essence of the problem and get insights into its solution. At the beginning of every time period $t$, the service takes the decision (action/control) to either keep the monitoring state the same as the previous time period or to switch to the other monitoring state. Formally, the action space is $\Acal=\{o,i\}$ and each (state, action) pair is associated with a cost given by the function $c:\Scal\times\Acal\mapsto\mathbb{R}^+$. The transition probabilities are given by $p(s'|s,a)$ where $s',s\in\Scal$ and $a\in\Acal$. 

We explain the evolution in the two-dimensional plane (i.e., $n=2$). This helps us simplify the notation and better visualize the optimal control (Section \ref{sec:opti-control}). The health state is denoted as $\hh=(h^{(x)},h^{(y)})$ and the system state is given by $(m,h^{(x)},h^{(y)})$, where $m\in\Mcal$. The definitions can easily be extended to higher dimensions. The cost function and transition probabilities are given below. For simplicity, we only discuss the states lying inside (and not on) the boundary here. The complete transition probabilities can be understood using Figure \ref{fig:Markov-transitions} and have been defined in Appendix \ref{app:boundary}.

\noindent\textbf{1. At critical health states $\hh\in\Hcal_C$ ---} 
\begin{quote}
    No action is taken with the service ceasing operation. A cost of $C_c$ is incurred.
\end{quote}

\noindent \textbf{2. When $\hh\notin\Hcal_C$ and $1\leq h^{(x)},h^{(y)}<H$ ---}
\begin{enumerate}[label=(\alph*)]
\item {\em Ordinary Monitoring $(m=o)$, no Switching $(a=o)$:} \\
Does not induce a monitoring change. Starting at state $(o,h^{(x)},h^{(y)})$, the next state is
\begin{enumerate}[label=\roman*)]
\item 
$(o,h^{(x)}+1,h^{(y)})$ with probability $\lambda_{o,x}$,
\item
$(o,h^{(x)},h^{(y)}+1)$ with probability $\lambda_{o,y}$,
\item $(o,h^{(x)}-1,h^{(y)})$ with probability $\mu_{o,x}$,
\item $(o,h^{(x)},h^{(y)}-1)$ with probability $\mu_{o,y}$,
\end{enumerate}
and a cost $C_o$ is incurred. 
\item {\em Intensive Monitoring $(m=i)$, no Switching $(a=i)$:}\\
Does not induce a monitoring change. Starting at state $(i,h^{(x)},h^{(y)})$, the next state is:
\begin{enumerate}[label=\roman*)]
\item 
$(i,h^{(x)}+1,h^{(y)})$ with probability $\lambda_{i,x}$,
\item
$(i,h^{(x)},h^{(y)}+1)$ with probability $\lambda_{i,y}$,
\item $(i,h^{(x)}-1,h^{(y)})$ with probability $\mu_{i,x}$,
\item $(i,h^{(x)},h^{(y)}-1)$ with probability $\mu_{i,y}$,
\end{enumerate}
and a cost $C_i$ is incurred.

\item {\em Intensive Monitoring $(m=i)$, with Switching $(a=o)$:}\\
Induces a switch to ordinary monitoring. The next state, respective transition probabilities, and the cost incurred are same as part (a): ordinary monitoring ($m=o$) with no switching $(a=o)$.

\item {\em Ordinary Monitoring $(m=o)$, with Switching $(a=i)$:}\\
Induces a switch to intensive monitoring. The next state, respective transition probabilities, and the cost incurred are same as part (b): intensive monitoring ($m=i$) with no switching $(a=i)$.
\end{enumerate}
Here $\lambda_{o,x}+\lambda_{o,y}+\mu_{o,x}+\mu_{o,y}=1$ and $\lambda_{i,x}+\lambda_{i,y}+\mu_{i,x}+\mu_{i,y}=1$ We assume that the decision to switch (or not) is made at the beginning of the time period. As a result, the transition probabilities and costs are decided solely by the next monitoring state. We make the following assumptions.
\begin{assumption} The transition probabilities and costs satisfy:
\begin{enumerate}[label=\alph*)]
    \item  $\lambda_{i,x}\geq\lambda_{o,x}$ and $\lambda_{i,y}\geq\lambda_{o,y}$.
    \item  $0\leq C_o\leq C_i\leq C_c$.
\end{enumerate}
\end{assumption}
Both of these assumptions are natural in practice. The first assumption 1.a) intuitively states that in comparison to ordinary monitoring, the patient's health improves faster under intensive monitoring. Assumption 1.b) is also expected as the patient has a lower quality of life under intensive monitoring. Further, given the severity of reaching a critical state, it is expected that $C_c$ is {\em much larger} than $C_i$ ($C_c\gg C_i$).

\begin{figure}[h!]
\centering
\includegraphics[width=0.7\linewidth]{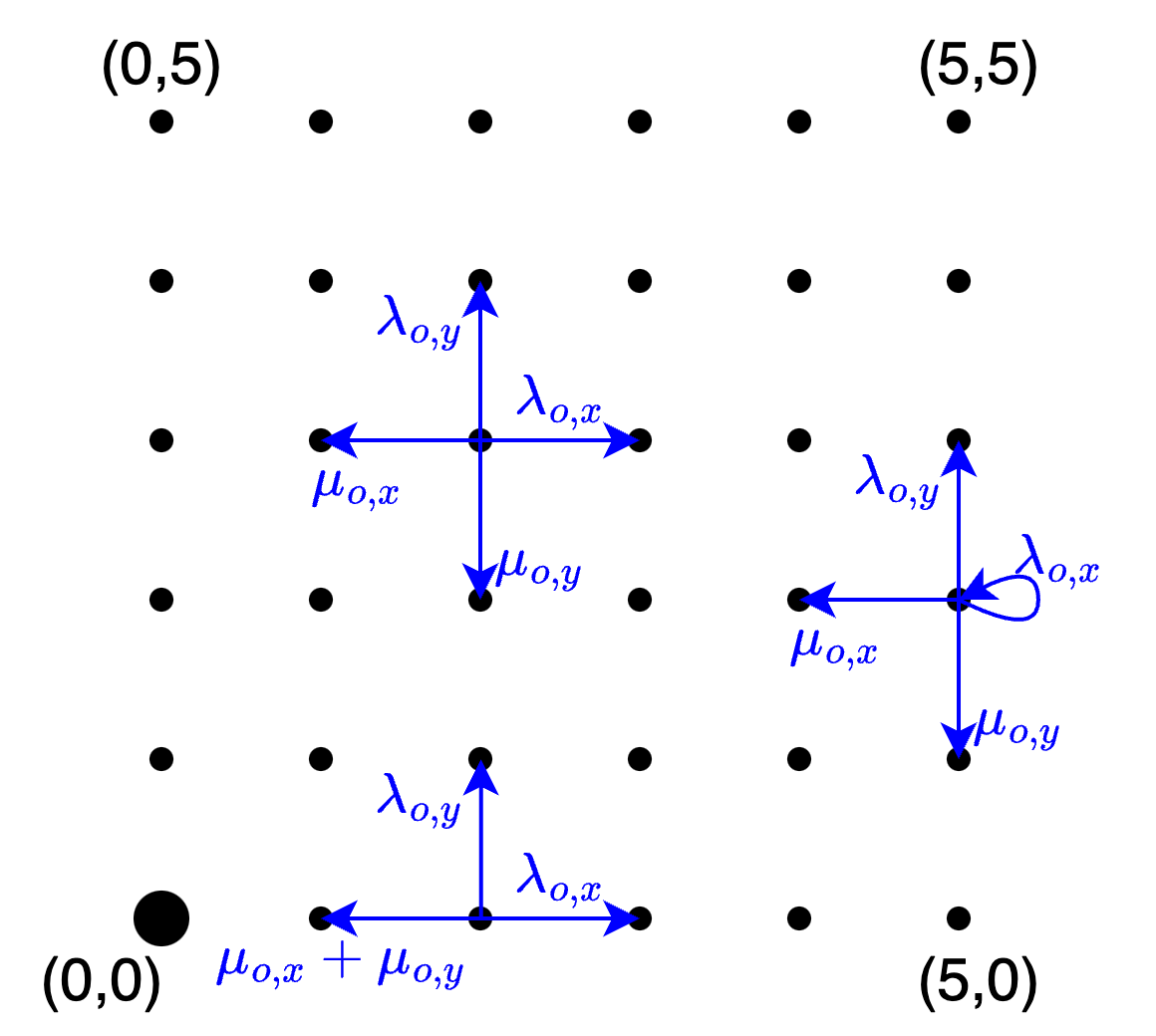}
\caption{The evolution for a 2-dimensional model with $H=5$ and origin as the critical health state. The transitions are marked for states $(o,2,0), (o,5,2)$, and $(o,2,3)$ under action $o$. The transaction probabilities have been labeled on the arrows, and the cost $C_o$ is incurred in each case.}
  \label{fig:Markov-transitions}
\end{figure}

There can be multiple generalizations for the above model. We can incorporate state-dependent costs and probabilities, and study tiers of progressively more intensive monitoring. Due to limited space, we assume constant costs, probabilities, etc., and explore the solution under these simplifications.

\subsection{Optimal Monitoring Control}
We study this problem under the {\em discounted cost} setting of the dynamic programming methodology \cite{Bertsekas}.  Costs incurred $t$ time periods into the future are discounted by a factor of $\gamma^t$ with $0<\gamma<1$. Starting from state $s_0=s\in\Mcal\times\Hcal$, the total expected discounted cost incurred is
\begin{equation*}
\EE\Big[\sum_{t=0}^{T-1} \gamma^tc(s_t,a_t)+\gamma^TC_c \Big | \ s_o=s\Big].    
\end{equation*}
At each time $t$, control $a_t$ is taken while in state $s_t=(m_t,\hh_t)$, and a cost of $c(s_t,a_t)$ is incurred. The service ceases operation when the patient reaches a critical state at time $T$, and the critical cost of $C_c$ is incurred, discounted to $\gamma^TC_c$. Thus, discounting by $\gamma$ implicitly reflects the patient's desire to stay away from the critical health states for longer.

A stationary monitoring policy or control $\pi$ is a mapping from each state $s\in\Scal$ to a control $\pi(s)\in\Acal$ taken at that state. The value function $V_\pi(s)$ for policy $\pi$ is the total expected cost incurred when starting at state $s$ and where the control at each step is decided using policy $\pi$, i.e., 
$$
V_\pi(s)=\EE\left[\sum_{t=0}^{T-1}
\gamma^t c\Big(s_t,\pi(s_t)\Big)
+\gamma^TC_c\ \Big |\ s_0=s\right].
$$
$V_{\pi}(s)$ satisfies the following dynamic programming equation \cite{Bertsekas} for all $s\in\Scal$.
$$
V_\pi(s)=c\Big(s,\pi(s)\Big)+
\gamma
\sum_{s'\in\Scal} \PP\Big(s'\mid s,\pi(s)\Big) V_\pi(s').
$$

Our goal is to find the optimal policy $\pi^*$ which minimizes $V_\pi$ over all policies $\pi$, i.e., $V_{\pi^*}(s) \le V_\pi(s)$ for every $s\in\Scal$ over all policies $\pi$. We define $V^*(s)\coloneqq V_{\pi^*}(s)$ which satisfies the following dynamic programming equation \cite{Bertsekas}.
$$
\Vstar(s)=\min_{a\in\{o,i\}} \Bigg\{c(s,a)+\gamma
\sum_{s'\in\Scal}\PP\Big(s'|s,a\Big)\Vstar(s')\Bigg\},
$$
and can be solved numerically to yield the optimal policy $\pistar$, which denotes the mapping from each state $s$ to the optimal decision at that state $s$. For the two-dimensional model defined previously, the above equation unfolds into:

\begin{enumerate}
[align=left, 
leftmargin=0pt, 
labelindent=\parindent, listparindent=\parindent, 
labelwidth=0pt, 
itemindent=!,
label=(\roman*)]
\item At critical health states $\hh\in\Hcal_C$:
\begin{equation*}
\Vstar(i,\hh)=\Vstar(o,\hh)=C_c.
\end{equation*}
\item 
For health states $\hh\notin\Hcal_C$, recall that the transition probabilities and cost incurred for a given time step is  decided by the action taken and the new monitoring state (and not the current one). This results in $\Vstar(i,\hh)$ and $\Vstar(o,\hh)$ being equal. For
$1\leq h^{(x)},h^{(y)}<H$, 
\begin{align}\label{dp-eqn}
    &\Vstar(i,\hh)=\Vstar(o,\hh)\nonumber\\
    &=\min\Bigg\{C_i+\gamma\bigg[\lambda_{i,x} \Vstar\left(i,h^{(x)}+1,h^{(y)}\right)\nonumber\\
    &\;\;+\lambda_{i,y}\Vstar\left(i,h^{(x)},h^{(y)}+1\right)+\mu_{i,x} \Vstar\left(i,h^{(x)}-1,h^{(y)}\right)\nonumber\\
    &\;\;+\mu_{i,y}\Vstar\left(i,h^{(x)},h^{(y)}-1\right)\bigg],\nonumber\\
    &\;\;C_o+\gamma\bigg[\lambda_{o,x} \Vstar\left(i,h^{(x)}+1,h^{(y)}\right)\nonumber\\
    &\;\;+\lambda_{o,y}\Vstar\left(i,h^{(x)},h^{(y)}+1\right)+\mu_{o,x} \Vstar\left(i,h^{(x)}-1,h^{(y)}\right)\nonumber\\
    &\;\;+\mu_{o,y}\Vstar\left(i,h^{(x)},h^{(y)}-1\right)\bigg]\Bigg\}
\end{align}
\end{enumerate}

For clarity, we focus on the states lying inside (and not on) the boundary in the above dynamic programming equation. The complete equation has been given in Appendix \ref{app:boundary}.

In the $\min\{\cdot,\cdot\}$ above, the first term stems from the control $i$, i.e., the patient is in intensive monitoring in the next timestep, and the second term stems from the control $o$.

\section{Critical Health States}
Critical health states are states where the patient's health has deteriorated significantly and the patient's health care is no longer within the scope of the current service model. Instead, other (emergency and more severe) medical interventions are required at that point and, hence, a large cost $C_c$ is incurred when the patient reaches these states. 

In our earlier model \cite{chandak2024tiered}, where the health states were one-dimensional, the only critical health state was the state $h=0$. But the set of critical health states $\Hcal_C$ can have diverse structures when the health states are multidimensional. As we discuss below, this set depends on the medical condition this model is applied to and is closely related to how different health measurements (dimensions of our model) interact. Choosing an appropriate set of critical health states is vital, as the optimal control is heavily influenced by this set. We discuss the optimal control further in the next section.

We consider $\Hcal_C$ of the form $\Hcal_C=\{\hh\mid g(\hh)\leq c\}$, where $g:\Hcal\mapsto\mathbb{R}$ and $c\in\mathbb{R}$. Intuitively, a health state is critical when the patient's health falls below a certain threshold, defined by a function of the health states $g(\cdot)$. The following are some sets which we consider in this short paper. In Figure \ref{fig:opti}, the critical sets are marked with larger black dots. 

\begin{enumerate}[label=(\alph*)]
    \item \textbf{Any health measurement is very low}: In this case, the patient's health is considered critical if even a single measurement becomes `very low' (normalized to 0 in our paper). For example, a person with diabetes experiencing a severe hypoglycemic episode would be considered to be in a critical state \cite{mccoy-hypo} even if they otherwise had a high time in range. Mathematically this can be represented as the set $\Hcal_C=\{\hh\mid \min_{1\leq i\leq n} \{h^{(i)}\}=0\}$. In the two-dimensional setting, this set would be the states along the $x$ and $y$ axes (Figure \ref{fig:opti-axes}).
    \item \textbf{Health measurements are collectively low}: A patient's health is considered critical in this case when the measurements are together  low. For example, when considering patients with severe cardiovascular disease, deterioration in cardiac function would be indicated by a combination of variables. While increased resting heart rate and decreased physical activity might not indicate severe risk separately, their combination with weight gain from water retention might suggest critical health \cite{Zinzuwadia, senanayake}. 
    
    This case can have multiple structures, but we consider the following two sets which correspond to the $\ell_1$ and $\ell_{\infty}$ norm, respectively - $\Hcal_C=\{\hh\mid \sum_{1\leq i\leq n} h^{(i)}\leq c\}$ and $\Hcal_C=\{\hh\mid \max_{1\leq i\leq n} \{h^{(i)}\}\leq c\}$. In the two-dimensional setting, these correspond to a triangle (Figure \ref{fig:opti-triangle}) and a square (Figure \ref{fig:opti-square}) cornered at the origin, respectively. Note that a critical set consisting of only the origin is a special case of this set.
    
    \item \textbf{Combination of the above sets}: A critical set applicable to a wide range of medical conditions is a combination of the above two sets. An example of such a set could be $\Hcal_C=\{\hh\mid \min_{1\leq i\leq n} \{h^{(i)}\}=0\ \textrm{or} \ \sum_{1\leq i\leq n} h^{(i)}\leq c\}$. In the two-dimensional setting, this set would include the axes, along with a triangle cornered at origin (Figure \ref{fig:opti-combi}). For example, a critical mental health state (such as high risk of suicide) may be indicated by several risk factors, and represented by such a critical health state \cite{batra, torous-suicide}.
\end{enumerate}
Note that these are just a few examples of the critical sets that can be represented using our general model. As we move to higher dimensions, there can be more complex sets which can be used to model medical conditions with a higher number of measurement types.

\section{Optimal control: The Switching Curve}\label{sec:opti-control}

\begin{figure*}[h!]
\centering
\begin{subfigure}[t]{0.24\textwidth}
  \centering
  \includegraphics[width=0.9\linewidth]{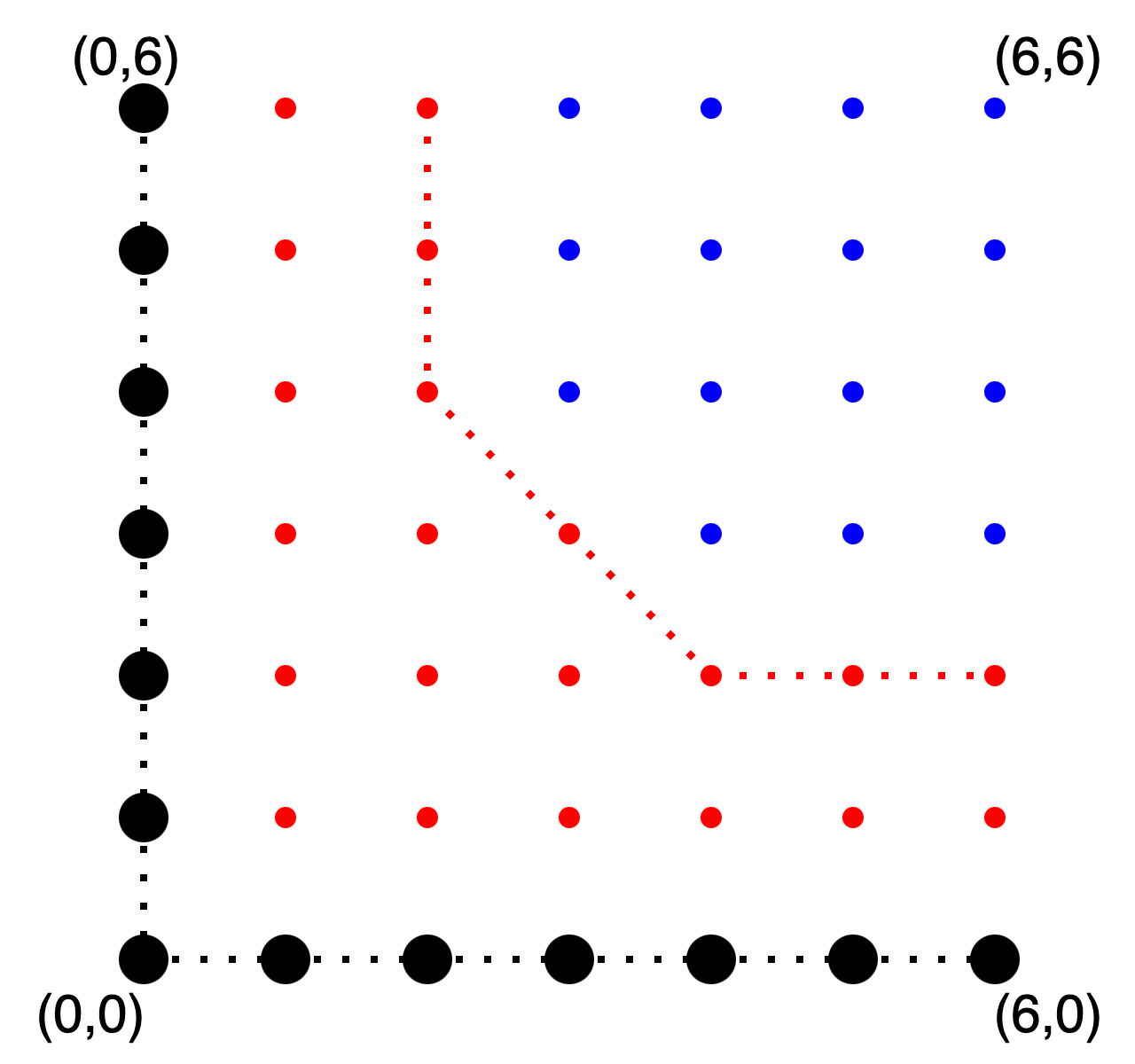}
  \caption{\centering The critical set comprises of the states on x- and y-axis.}
  \label{fig:opti-axes}
\end{subfigure}%
\hfill
\begin{subfigure}[t]{.24\textwidth}
  \centering
  \includegraphics[width=.9\linewidth]{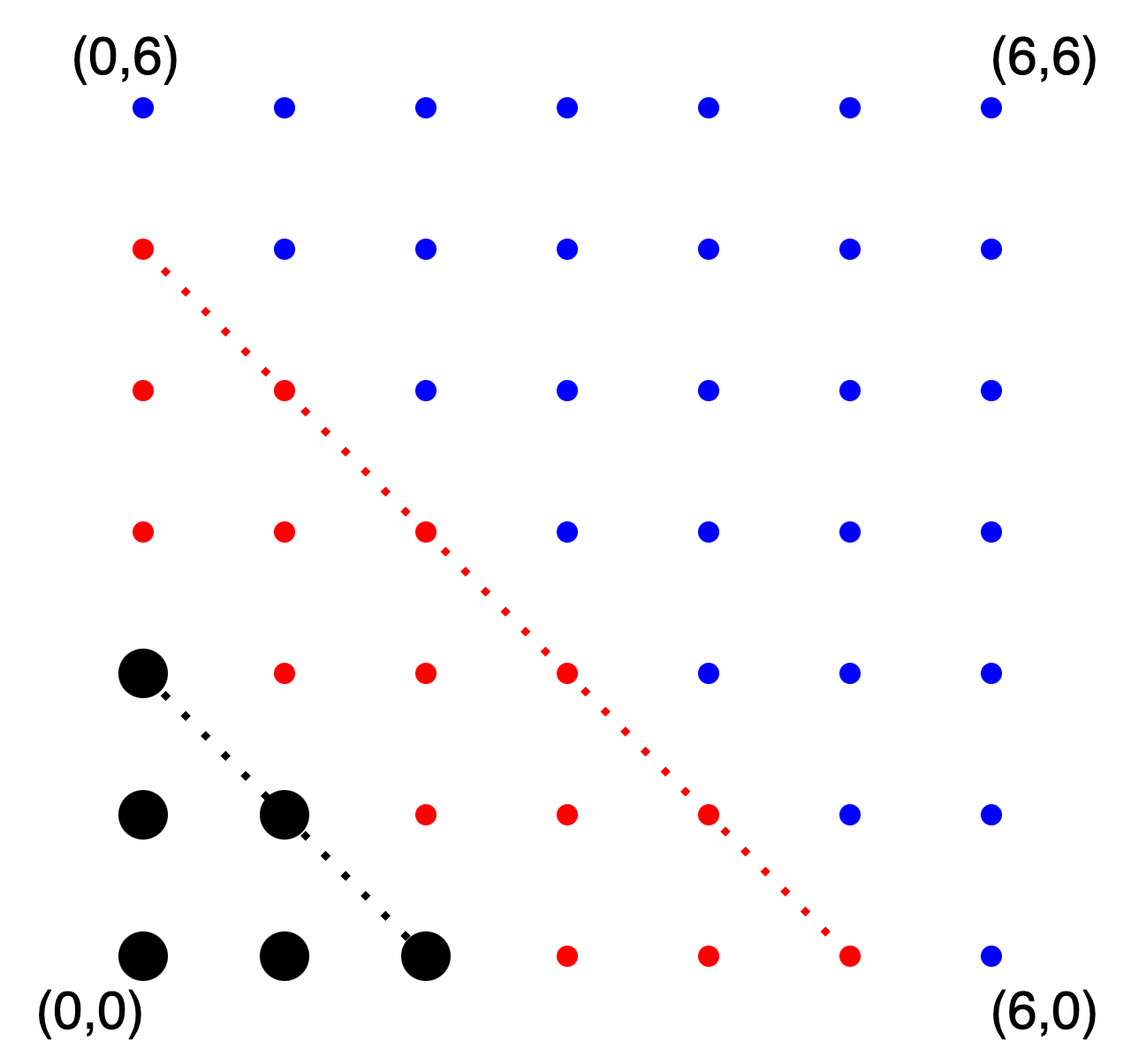}
  \caption{\centering The critical set comprises of states in a triangle around origin.}
  \label{fig:opti-triangle}
\end{subfigure}
\hfill
\begin{subfigure}[t]{.24\textwidth}
  \centering
  \includegraphics[width=.9\linewidth]{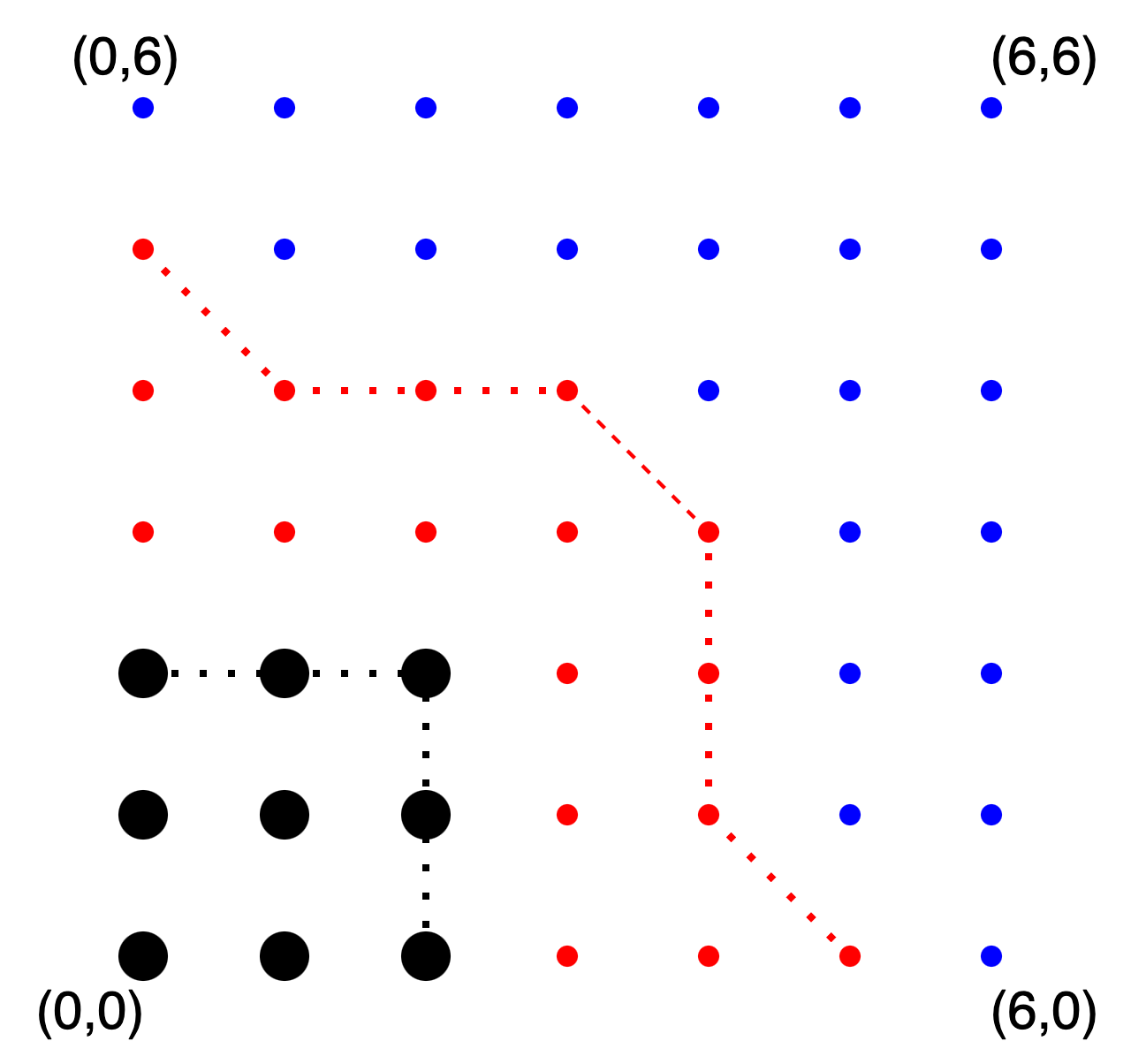}
  \caption{\centering The critical set comprises of states in a square around origin.}
  \label{fig:opti-square}
\end{subfigure}
\hfill
\begin{subfigure}[t]{.24\textwidth}
  \centering
  \includegraphics[width=.9\linewidth]{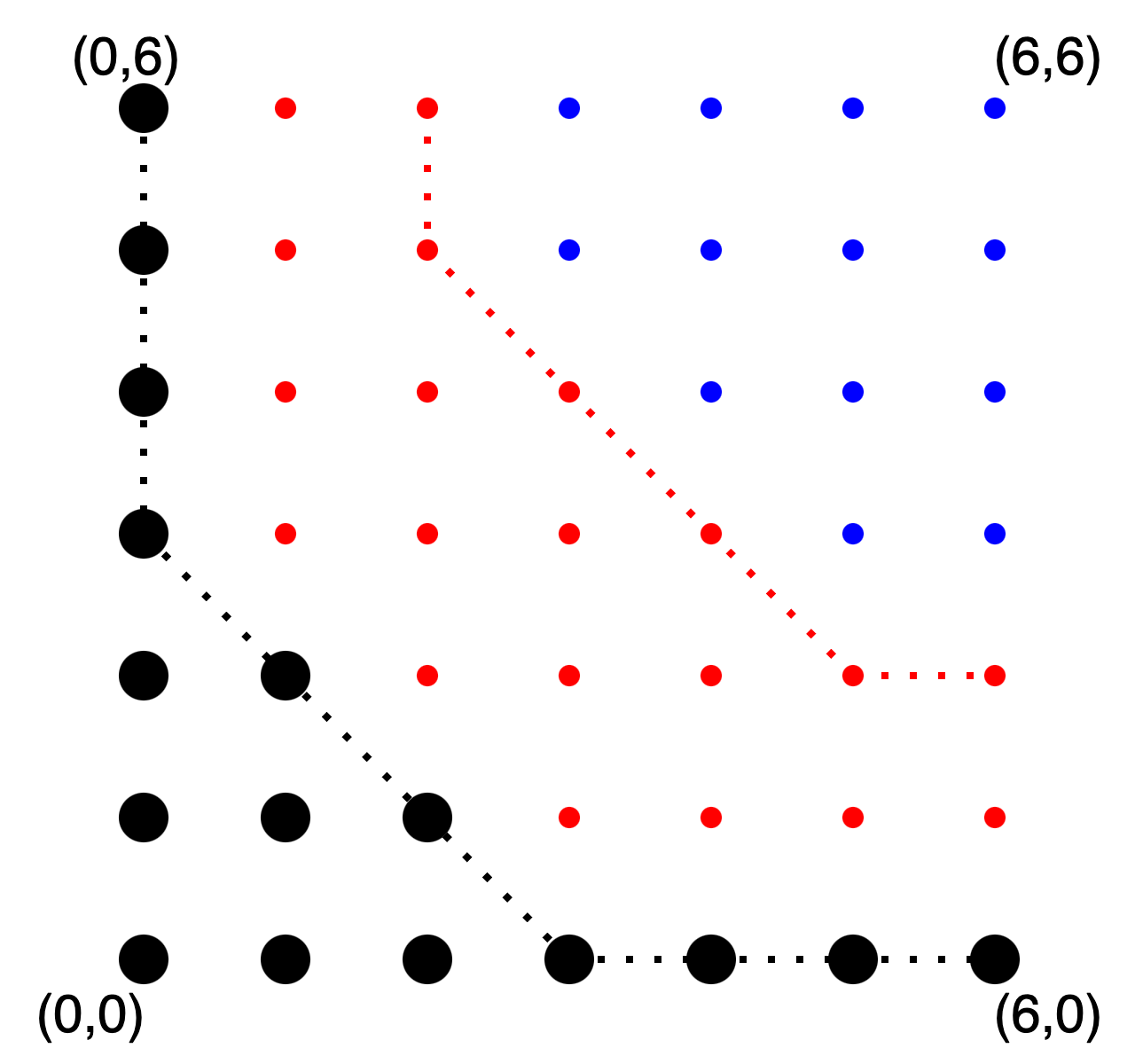}
  \caption{\centering The critical set comprises of union of sets in (a) and (b).}
  \label{fig:opti-combi}
\end{subfigure}
\caption{Optimal controls for a two-dimensional model with $H=6$ for critical sets given by (a) $h^{(x)}=0$ or $h^{(y)}=0$, (b) $h^{(x)}+h^{(y)}\leq 2$, (c) $\max\{h^{(x)},h^{(y)}\}\leq 2$, and (d) $h^{(x)}=0$ or $h^{(y)}=0$ or $h^{(x)}+h^{(y)}\leq 2$. Larger black dots represent critical health states, and dotted black line represents the boundary of the critical set. Intensive monitoring is optimal for states marked as red, and ordinary monitoring is optimal for states marked as blue. The dotted red line represents the {\it switching curve} $f(\hh)=0$. For each of the plots: $\gamma=0.9, C_i=1, C_o=0, C_c=35, \lambda_{o,x}=\lambda_{o,y}=0.5-\mu_{o,x}=0.5-\mu_{o,y}=0.075$ and $\lambda_{i,x}=\lambda_{i,y}=0.5-\mu_{i,x}=0.5-\mu_{i,y}=0.2$.}
\label{fig:opti}
\end{figure*} 

In this section, we discuss the optimal monitoring strategy for various critical health sets and model parameters. We use value iteration \cite{Bertsekas} to solve the dynamic programming equation \eqref{dp-eqn} and obtain the optimal policy. While we discuss some analytical results to help build intuition, various theoretical results are not discussed, given the limited space of this short paper. We postpone these results to a future work.

We first make an important observation about the optimal policy. We observe that $\Vstar(i,\hh)=\Vstar(o,\hh)$ for all health states $\hh$. This implies that $\pistar(o,\hh)=\pistar(i,\hh)$. Hence under the optimal control, the service chooses the same monitoring  irrespective of the current monitoring state.
Therefore, we introduce the notation $\pistar(\hh)=\pistar(i,\hh)=\pistar(o,\hh)$ denoting the optimal monitoring at state $\hh$. For example, suppose that $\pistar(\hh)=i$ for some health state $\hh$. If the current state is $(o,\hh)$, then the optimal control is to switch to intensive monitoring, and if the current state is $(i,\hh)$, then the optimal control is to stay in intensive monitoring. 

The above observation helps us better visualize the optimal control. In this section, we visualize the optimal control for a 2-dimensional health state model using a 2D grid. In Figures \ref{fig:opti} and \ref{fig:asymmetry}, the critical health states are marked with larger black dots. Ordinary and intensive monitoring are optimal for health sets marked with \textit{blue} and \textit{red}, respectively. 

Before moving to optimal control for specific instances, we make an important numerical observation. For our class of $\Hcal_C=\{\hh\mid g(\hh)\leq c\}$, we observe that the optimal control is a threshold policy $\pithresh$, characterized by function $f:\Hcal\mapsto\mathbb{R}$. Then $\pithresh(\hh)=i$ if $f(\hh)\leq 0$, and $\pithresh(\hh)=o$ if $f(\hh)>0$. For example, consider Figure \ref{fig:opti-triangle}. In that case,  $f(\hh)=h^{(x)}+h^{(y)}-5$ and hence the optimal control is $\pistar(\hh)=i$ if $h^{(x)}+h^{(y)}\leq 5$ and $o$ otherwise. A lower value for $f(\hh)$ implies that the patient is `closer' to the critical set, and hence intensive monitoring may be preferred for state $\hh$.  

The curve $f(\hh)=0$ is called the {\bf switching curve}. The monitoring control switches from ordinary to intensive as the patient's health state falls below the switching curve. The switching curve is seen in Figures \ref{fig:opti} and \ref{fig:asymmetry} as the red dotted boundary between the region of blue dots and the region of red dots. Note that in higher dimensions, the switching curve will be replaced by a switching hypersurface.

\subsection{Impact of Critical Health States}

We first discuss how the previously discussed  critical health sets affect the optimal policy. In Figure \ref{fig:opti}, we consider `symmetric' sets and parameters, i.e., the critical sets are symmetric in the two health measurements (dimensions) and the probabilities satisfy $\lambda_{o,x}=\lambda_{o,y}=0.5-\mu_{o,x}=0.5-\mu_{o,y}$ and $\lambda_{i,x}=\lambda_{i,y}=0.5-\mu_{i,x}=0.5-\mu_{i,y}$. 

The first case we discuss is the set where all points on the axes are considered critical states. Figure \ref{fig:opti-axes} shows the optimal control for this set.  Closer to the origin, we observe that the optimal control assigns intensive monitoring to more health states, as the patient is at higher risk of reaching either axis (and hence a critical state). Further away from the origin, the patient's risk of reaching a critical state can be characterized by its distance to just the closer axis, and hence the switching curve is parallel to the axis.

This can also be understood using the connection between hitting times and optimal policy studied for the one-dimensional system in \cite{chandak2024tiered}. Under the assumption of $H\uparrow \infty$, we showed that the function $\EE[\gamma^{\tau(h)}]$ is closely related to the optimal policy, where $\tau(h)$ denotes the hitting time for the critical health state starting at state $h$. In the present case, the time taken to hit a critical state is given by $\min\{\tau_x(\hh),\tau_y(\hh)\}$ where $\tau_x(\hh)$ and $\tau_y(\hh)$ denote the time taken to hit the x and the y-axis, respectively. We numerically observe that the level sets for $\EE[\gamma^{\min\{\tau_x(\hh),\tau_y(\hh)\}}]$ have a structure very similar to the optimal control we observe.

We next study the optimal control when the set of critical health states is a triangle cornered at origin, given by $h^{(x)}+h^{(y)}\leq c$. In this case (Figure \ref{fig:opti-triangle}), we observe that the switching curve has the same form, i.e., the patients are in intensive monitoring when $h^{(x)}+h^{(y)}\leq k$, for some $k$. In this case, the hitting time for the critical set (starting at state $\hh$) is a function of $(h^{(x)}+h^{(y)}-c)$ and hence the level sets are of the form $h^{(x)}+h^{(y)}=k$. We have the following result for the asymptotic case which shows that the optimal control in this case will always be of this form. 
\begin{theorem}\label{thm:triangle}
    Consider critical sets of the form $\Hcal_C=\{\hh\mid h^{(x)}+h^{(y)}\leq c\}$. Then, under the assumption that $H\uparrow\infty$ and for sufficiently small $\gamma$, the optimal control is a threshold policy $\pithresh$, where the switching curve is $f(\hh)=h^{(x)}+h^{(y)}-k$ for some $k\geq c$.
\end{theorem}
\begin{proof}
    See Appendix \ref{app:proofs} for proof sketch.
\end{proof}
Note that this result holds for any set of probabilities, and does not require them to be symmetric. The proof for this theorem follows from looking at this two-dimensional grid as a one-dimensional random walk, where all health states in the set $A^{(k)}=\{\hh\mid h^{(x)}+h^{(y)}=k\}$ are considered as state $k$ of the one-dimensional random walk. 

We next study the critical set given by the combination of the above two sets, i.e., when all states on the axes and states in a triangle around the origin are critical health states. Figure \ref{fig:opti-combi} gives the optimal control for this set. The shape of the switching curve is very similar to that in Figure \ref{fig:opti-axes}, but intensive monitoring is optimal for a larger region around the origin. Finally we present the optimal control for the critical health state corresponding to $\max\{h^{(x)}, h^{(y)}\}\leq c$ (Figure \ref{fig:opti-square}). The switching curve in this case is structurally more complex than the previous cases. Even though the structure is complex, different pieces of the switching curve can intuitively be explained based on the time taken to reach a critical state starting from any given state.

\subsection{Asymmetry Between Dimensions}
In the previous subsection, we only considered `symmetric' critical sets and parameters. Now we discuss how the optimal control is affected when the model is asymmetric. 

In Figure \ref{fig:asymmetry-axes}, we plot the optimal control for the two-dimensional model where health states along the axes are critical, and $\lambda_{i,x}-\lambda_{o,x}>\lambda_{i,y}-\lambda_{o,y}$. This implies that, under intensive monitoring, the probability of 
patient's health measurement $h^{(x)}$ improving is higher than that for $h^{(y)}$. In this case, the patient has a high incentive to pay the cost for intensive monitoring even when they are at `mild' risk of reaching the critical state corresponding to $h^{(x)}=0$. On the other hand, the advantage of intensive monitoring for measurement $h^{(y)}$ (in terms of improvement probability) is lower, and hence the patient has an incentive to pay the cost for intensive monitoring only when their $h^{(y)}$ is very low.  

\begin{figure}[h!]
\centering
\begin{subfigure}[t]{0.24\textwidth}
  \centering
  \includegraphics[width=0.9\linewidth]{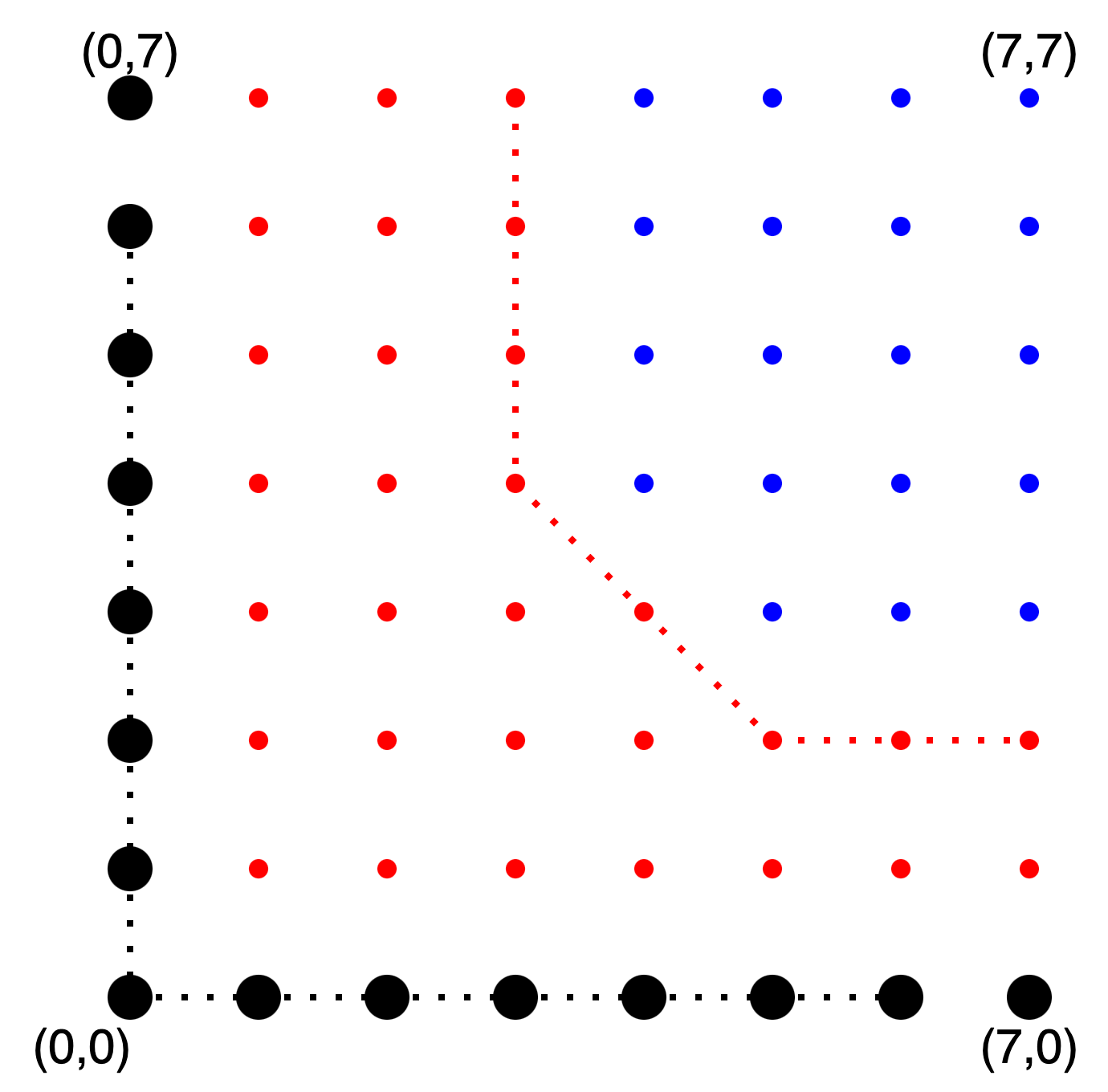}
  \caption{\centering The critical set comprises of the states on x- and y-axis.}
  \label{fig:asymmetry-axes}
\end{subfigure}%
\begin{subfigure}[t]{.245\textwidth}
  \centering
  \includegraphics[width=.9\linewidth]{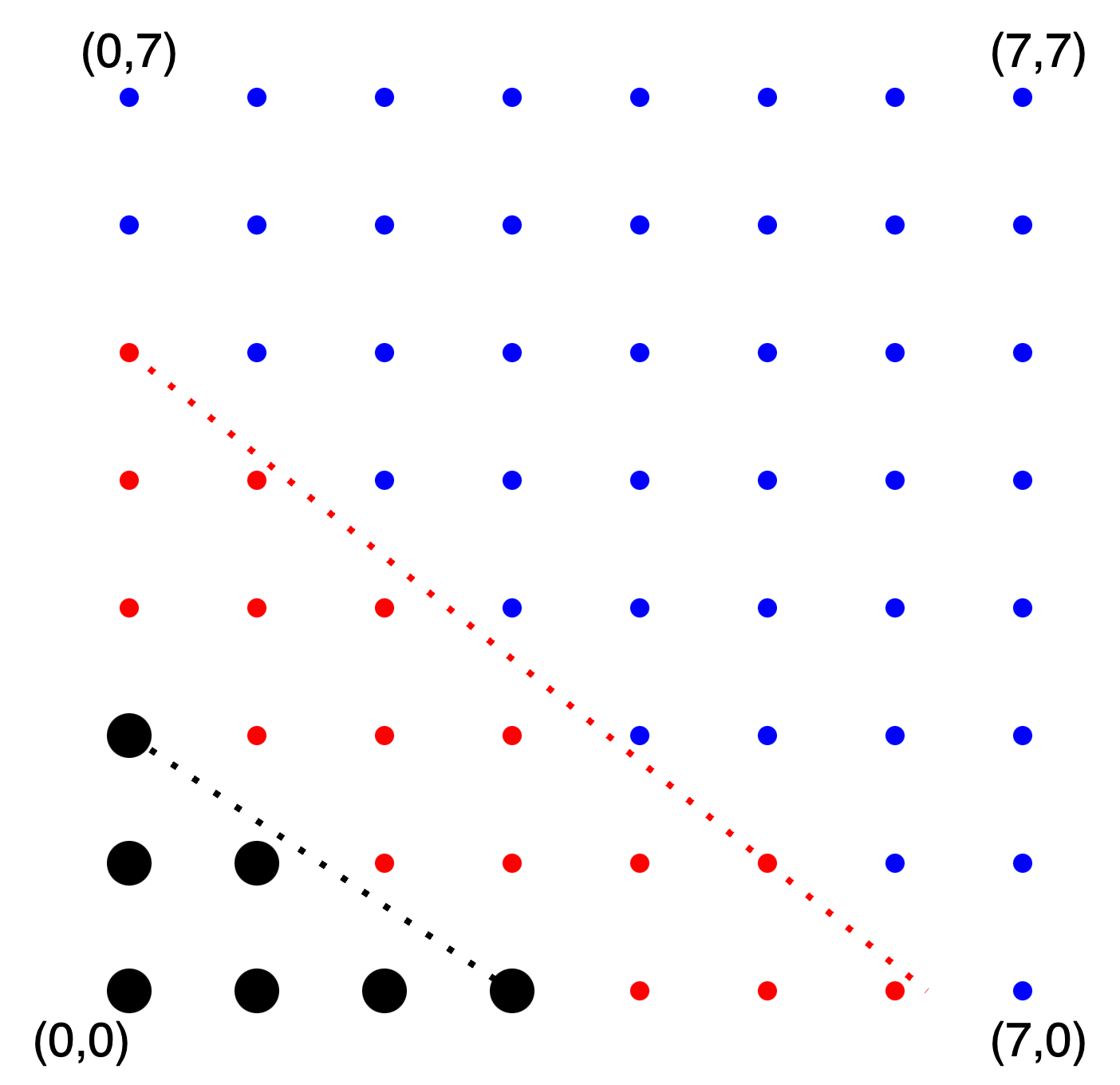}
  \caption{\centering The critical set comprises of states in a triangle around origin}
  \label{fig:asymmetry-triangle}
\end{subfigure}
\caption{Optimal control for \textbf{asymmetric} two-dimensional models for critical sets given by (a) $h^{(x)}=0$ or $h^{(y)}=0$, and (b) $2h^{(x)}+3h^{(y)}\leq 6$. For both plots: $\gamma=0.9, C_i=1, C_o=0, C_c=35$. For plot (a): $\lambda_{o,x}=\lambda_{o,y}=0.5-\mu_{o,x}=0.5-\mu_{o,y}=0.1$ and $\lambda_{i,x}=0.5-\mu_{i,x}=0.3$ and $
\lambda_{i,y}=0.5-\mu_{i,y}=0.25$. For plot (b): $\lambda_{o,x}=\lambda_{o,y}=0.5-\mu_{o,x}=0.5-\mu_{o,y}=0.1$ and $\lambda_{i,x}=\lambda_{i,y}=0.5-\mu_{i,x}=0.5-\mu_{i,y}=0.2$}
\label{fig:asymmetry}
\end{figure} 
The next plot (Figure \ref{fig:asymmetry-triangle}) considers an asymmetric critical set. Here $\Hcal_C=\{\hh\mid 2h^{(x)}+3h^{(y)}\leq 6\}$. In this case, switching curve is $f(\hh)=4h^{(x)}+5h^{(y)}-25=0$. Note that the slope for the switching curve $f(\cdot)$ is not the same as the slope for the critical set function $g(\cdot)$. But we empirically observed that the slopes of functions $g(\cdot)$ and $f(\cdot)$ are close under most cases for a large set of asymmetric critical sets.

\subsection{Variation With Parameters}
Finally, we report the impact of different model parameters on the optimal control. The observations here are very similar to that in the one-dimensional case \cite{chandak2024tiered}. Due to the limited space in this paper, we do not include plots for these results. 

Increasing the discount factor $\gamma$, increasing the cost ratio of $C_c/C_i$ (keeping $C_o$ fixed at zero), or increasing the probabilities $\lambda_{i,x}$ or $\lambda_{i,y}$ (keeping $\lambda_{o,x}$ and $\lambda_{o,y}$ as fixed) have similar effects on the optimal control. Each of these \textit{push} the switching curve away from the origin, with the patient staying under intensive monitoring for a larger set of health sets. In the first two cases, the patient incurs a higher discounted cost on reaching the critical state, and hence the patient has a higher incentive to stay in intensive monitoring. In the last case, the probability of
the patient’s health improving under intensive monitoring improves, incentivizing the patient to stay under intensive monitoring for longer. 

\section{Conclusions}\label{sec:conclusions}
We introduce a multidimensional health state model, extending our prior work \cite{chandak2024tiered} which considered one-dimensional health states. The monitoring service control decides whether to place the patient under ordinary or intensive monitoring, given their health state. Optimal monitoring control is then studied using a dynamic programming approach. Our observations show that the optimal control, is characterized by a \textit{switching curve}, such that patients transition to intensive monitoring when their health is below the switching curve. 

A significant next step involves refinement of our framework using real-world RPM data. In practice, domain experts could define the health states by leveraging existing, clinically relevant metrics (e.g., glucose ranges). Patient-level data collected from remote monitoring devices could be collected in a pilot study to inform the estimation of transition probabilities, while cost parameters may be derived from published literature, actual monetary expenditures in clinics, or expert opinion regarding the burden on patients’ quality of life. Beyond optimizing individual patient care, our model offers insights for healthcare system-level planning, such as estimating clinician staffing needs. Nonetheless, generalizing transition probabilities and cost structures beyond a single clinic or patient population remains a challenge: variations in demographics, comorbidities, or adherence behaviors can limit external validity. Additionally, RPM programs and technologies often evolve rapidly—new devices or improved protocols may alter how patients transition between health states and respond to monitoring. Consequently, periodic re-estimation of model parameters and reassessment of cost structures become critical for maintaining accuracy.

Another important direction would be to tackle the problems of the `curse of dimensionality' and the `curse of non-Markovianity' \cite{non-Markov}. The total number of health states increase exponentially with increase in the dimension of the health states.
This makes solving the dynamic programming equation in \eqref{dp-eqn} intractable when dealing with a large number of health measurements. In such cases, one can instead work with heuristic-based approximations of the optimal policy, or fasten the convergence of the value iteration algorithm by choosing `good' estimates for the value function $V(\cdot)$ as the starting point for the iteration \cite{Dua}. The `curse of non-Markovianity' arises when the observed health measurements are non-Markovian. Our dynamic programming methodology cannot be applied in such cases. This can either be tackled by defining heuristic-based health states which are assumed to satisfy the Markovian property \cite{Schaefer2004}, or by combining the value iteration algorithm with an algorithm which learns Markovian embeddings for the observations \cite{non-Markov}.




\appendices

\section{Mathematical Formulation and Proofs}
\subsection{Transition Probabilities and DP Equation}\label{app:boundary}
We first define the exact transition probabilities for our model. To account for the upper boundary, e.g., $h^{(x)}=H$, we replace the term $h^{(x)}+1$ with $\min\{h^{(x)}+1,H\}$ for our next state. And to account for the lower boundary, e.g., $h^{(x)}=0$, we use $\indi_{h^{(x)}\neq 0}$ and $\indi_{h^{(x)}= 0}$ for transition probabilities in $x$ and $y$-direction, respectively.

\noindent\textbf{1. At critical health states $\hh\in\Hcal_C$ ---} 
\begin{quote}
    No action is taken with the service ceasing operation. A cost of $C_c$ is incurred.
\end{quote}

\noindent \textbf{2. When $\hh\notin\Hcal_C$ and $0\leq h^{(x)},h^{(y)}\leq H$ ---}
\begin{enumerate}[label=(\alph*)]
\item {\em Ordinary Monitoring (m=o), no Switching (a=o):} \\
Starting at state $(o,h^{(x)},h^{(y)})$, the next state with their respective transition probabilities are:
\begin{enumerate}[label=\roman*)]
\item 
$(o,\min\{h^{(x)}+1,H\},h^{(y)})$ w.p. $\lambda_{o,x}$,
\item
$(o,h^{(x)},\min\{h^{(y)}+1,H\})$ w.p. $\lambda_{o,y}$,
\item $(o,h^{(x)}-1,h^{(y)})$ w.p. $\mu_{o,x}\indi_{\{h^{(x)}\neq 0\}}+\mu_{o,y}\indi_{\{h^{(y)}= 0\}}$,
\item $(o,h^{(x)},h^{(y)}-1)$ w.p. $\mu_{o,y}\indi_{\{h^{(y)}\neq 0\}}+\mu_{o,x}\indi_{\{h^{(x)}= 0\}}$.
\end{enumerate}

\item {\em Intensive Monitoring (m=i), no Switching (a=i):}\\
Starting at state $(i,h^{(x)},h^{(y)})$, the next state with their respective transition probabilities are:
\begin{enumerate}[label=\roman*)]
\item 
$(i,\min\{h^{(x)}+1,H\},h^{(y)})$ w.p. $\lambda_{i,x}$,
\item
$(i,h^{(x)},\min\{h^{(y)}+1,H\})$ w.p. $\lambda_{i,y}$,
\item $(i,h^{(x)}-1,h^{(y)})$ w.p. $\mu_{i,x}\indi_{\{h^{(x)}\neq 0\}}+\mu_{i,y}\indi_{\{h^{(y)}= 0\}}$,
\item $(i,h^{(x)},h^{(y)}-1)$ w.p. $\mu_{i,y}\indi_{\{h^{(y)}\neq 0\}}+\mu_{i,x}\indi_{\{h^{(x)}= 0\}}$.
\end{enumerate}

\item {\em Intensive Monitoring (m=i), with Switching (a=o):}\\
Same as part (a).

\item {\em Ordinary Monitoring (m=o), with Switching (a=i):}\\
Same as part (b).
\end{enumerate}

Next, we give the dynamic programming equations satisfied by the optimal control $\Vstar(\cdot,\cdot)$.

\noindent\textbf{1. At critical health states $\hh\in\Hcal_C$ ---} 
\begin{equation*}
\Vstar(i,\hh)=\Vstar(o,\hh)=C_c.
\end{equation*}

\noindent \textbf{2. When $\hh\notin\Hcal_C$ and $0\leq h^{(x)},h^{(y)}\leq H$ ---}
\begin{align*}
    &\Vstar(i,\hh)=\Vstar(o,\hh)\nonumber\\
    &=\min\Bigg\{C_i+\gamma\bigg[\lambda_{i,x} \Vstar\left(i,\min\{h^{(x)}+1,H\},h^{(y)}\right)\nonumber\\
    &+\lambda_{i,y}\Vstar\left(i,h^{(x)},\min\{h^{(y)}+1,H\}\right)\\
    &+\left(\mu_{i,x}\indi_{\{h^{(x)}\neq 0\}}+\mu_{i,y}\indi_{\{h^{(y)}= 0\}}\right) \Vstar\left(i,h^{(x)}-1,h^{(y)}\right)\nonumber\\
    &+\left(\mu_{i,y}\indi_{\{h^{(y)}\neq 0\}}+\mu_{i,x}\indi_{\{h^{(x)}= 0\}}\right)\Vstar\left(i,h^{(x)},h^{(y)}-1\right)\bigg],\nonumber\\
    &C_o+\gamma\bigg[\lambda_{o,x} \Vstar\left(o,\min\{h^{(x)}+1,H\},h^{(y)}\right)\nonumber\\
    &+\lambda_{o,y}\Vstar\left(o,h^{(x)},\min\{h^{(y)}+1,H\}\right)\\
    &+\left(\mu_{o,x}\indi_{\{h^{(x)}\neq 0\}}+\mu_{o,y}\indi_{\{h^{(y)}= 0\}}\right) \Vstar\left(o,h^{(x)}-1,h^{(y)}\right)\nonumber\\
    &+\left(\mu_{o,y}\indi_{\{h^{(y)}\neq 0\}}+\mu_{o,x}\indi_{\{h^{(x)}= 0\}}\right)\Vstar\left(o,h^{(x)},h^{(y)}-1\right)\bigg]\Bigg\}.
\end{align*}

\subsection{Proof for Theorem 1}\label{app:proofs}

\begin{proof}
We work under the asymptotic condition of $H\uparrow\infty$ for this proof. 
Recall that the one-dimensional model in \cite{chandak2024tiered} considered health state $h=0$ as the critical set and defined the parameters $\gamma, \lambda_o$ and  $\lambda_i$. These are the discount factor, and probability of health improving under ordinary and intensive monitoring, respectively. Theorem 1 and 2 from \cite{chandak2024tiered} together show that for sufficiently small $\gamma$, the optimal policy is always a threshold policy, i.e., there exists $\bar{h}$, such that $\pistar(h)=i$ for $h\leq \bar{h}$ and $\pistar(h)=o$ for $h> \bar{h}$. Note that the policy where the control at all states is ordinary monitoring is a special case with $\bar{h}=0$. 

Now in our two-dimensional model, 
consider the sets $A^{(k)}=\{\hh\mid h^{(x)}+h^{(y)}=k\}$. Then $\PP(\hh_{t+1}\in A^{(k+1)}\mid \hh_{t}\in A^{(k)}, m_t=o)=\lambda_{o,x}+\lambda_{o,y}$. Similarly, $\PP(\hh_{t+1}\in A^{(k-1)}\mid \hh_{t}\in A^{(k)}, m_t=o)=\mu_{o,x}+\mu_{o,y}$. Similarly, transitions and respective probabilities are defined for intensive monitoring. Let $\lambda_i'=\lambda_{i,x}+\lambda_{i,y}$ and $\lambda_o'=\lambda_{o,x}+\lambda_{o,y}$. Suppose the set of health sets $A^{(c)}=\{\hh\mid h^{(x)}+h^{(y)}=c\}$ is defined as the health set $h'=0$. Then sets of health states $A^{(k)}$ are given by $h'=k-c$ for $k\geq c$. Then our two-dimensional model can be represented using the one-dimensional model with parameters $\gamma,\lambda'_o,\lambda_i'$ and with health states given by $h'$. Then applying Theorems 1 and 2 from \cite{chandak2024tiered} gives us the result that the optimal control in the one-dimensional case is a threshold policy when $\gamma$ is sufficiently small. Let that threshold in the one-dimensional case be $\bar{h'}$, then the optimal control in the two-dimensional case is $\pithresh$ where $f(\hh)=h^{(x)}+h^{(y)}-(\bar{h'}+c)$. This completes the proof for Theorem 1.
\end{proof}

\end{document}